%% file: main.tex
\documentclass[11pt]{article}

\usepackage{algorithm}
\usepackage{amssymb,amsmath,amscd,latexsym,epic,eepic,psfrag,psfig,graphicx}
\usepackage{luca} 
\usepackage{macro} 
\usepackage{fullpage}

\newtheorem{definition}{Definition}
\newtheorem{theorem}{Theorem}
\newtheorem{lemma}{Lemma}
\newtheorem{corollary}{Corollary}
\newtheorem{example}{Example}

\title{Strategy Improvement for Concurrent Safety Games\thanks{This
research was suppored in part by the NSF grants
CCR-0132780, CNS-0720884, and CCR-0225610, and by the Swiss
National Science Foundation.}
}

\author{Krishnendu Chatterjee$^\S$ \qquad
  Luca de Alfaro$^{\S}$ \qquad
  Thomas A. Henzinger$^{\dag\ddag}$\\[5pt]
\normalsize
  $\strut^\S$ CE, University of California, Santa Cruz,USA\\
\normalsize
  $\strut^\dag$ EECS,
  University of California, Berkeley,USA\\
\normalsize
  $\strut^\ddag$ Computer and Communication Sciences, EPFL, Switzerland\\
\normalsize
  \texttt{$\{$c\_krish,tah$\}$@eecs.berkeley.edu, luca@soe.ucsc.edu}
}

\date{
}

\begin{document}

\maketitle
\begin{abstract}
We consider concurrent games played on graphs.  At every round of the
game, each player simultaneously and independently selects a move;
the moves jointly determine the transition to a successor state. 
Two basic objectives are the safety objective: ``stay forever in a set $F$ of
states'', and its dual, the reachability objective, ``reach a set $R$ of
states''.
We present in this paper a strategy improvement algorithm for computing the
\emph{value} of a concurrent safety game, that is, the maximal probability 
with which player~1 can enforce the safety objective. 
The algorithm yields a sequence of player-1 strategies which ensure
probabilities of winning that converge monotonically to the value of
the safety game. 

The significance of the result is twofold. 
First, while strategy improvement algorithms were known for 
Markov decision processes
and turn-based games, as well as for concurrent reachability games, 
this is the first strategy improvement algorithm for concurrent
safety games.
Second, and most importantly, the improvement algorithm provides
a way to approximate the value of a concurrent safety game 
{\em from below\/}
(the known value-iteration algorithms approximate the value from above).
Thus, when used together with value-iteration algorithms, or with
strategy improvement algorithms for reachability games, our algorithm
leads to the first practical algorithm for computing converging upper
and lower bounds for the value of reachability and safety games.
\end{abstract}

\section{Introduction}

We consider games played between two players on graphs. 
At every round of the game, each of the two players selects a move; the
moves of the players then determine the transition to the successor
state. 
A play of the game gives rise to a path on the graph. 
We consider two basic goals for the players: 
{\em reachability,} and {\em safety.} 
In the reachability goal, player~1 must reach a set of target states
or, if randomization is needed to play the game, then player~1 must
maximize the probability of reaching the target set. 
In the safety goal, player~1 must ensure that a set of target states
is never left or, if randomization is required, then player~1 must ensure that
the probability of leaving the target set is as low as possible. 
The two goals are dual, and the games are determined: the maximal
probability with which player~1 can reach a target set is equal to one
minus the maximal probability with which player~2 can confine the game in the
complement set \cite{Martin98}.

These games on graphs can be divided into two classes: 
{\em turn-based\/} and {\em concurrent.}
In turn-based games, only one player has a choice of moves at each
state; 
in concurrent games, at each state both players choose a move,
simultaneously and independently, from a set of available moves. 

For turn-based games, the solution of games with reachability and
safety goals has long been known. 
If the move played determines uniquely the successor state, the games
can be solved in linear-time in the size of the game graph.
If the move played determines a probability distribution over the
successor state, the problem of deciding whether a safety of
reachability can be won with probability greater than $p \in [0,1]$
is in NP $\cap$ co-NP \cite{Con92}, and the exact value of a game
can be computed by strategy improvement algorithms \cite{Con93}. 
These results all hinge on the fact that turn-based reachability and
safety games can be optimally won with deterministic, and memoryless,
strategies.  
These strategies are functions from states to moves, so they are
finite in number, and guarantees  the termination of the algorithms. 

The situation is different for the concurrent case, where
randomization is needed even in the case in which the moves played by
the players uniquely determine the successor state.
The {\em value\/} of the game is defined, as usual, as the sup-inf
value: the supremum, over all strategies of player~1, of the infimum,
over all strategies of player~2, of the probability of achieving the
safety or reachability goal. 
In concurrent reachability games, players are only guaranteed the existence of
$\varepsilon$-optimal strategies, that ensure that the value of the
game is achieved within a specified $\varepsilon > 0$ \cite{Kumar81};
these strategies (which depend on $\varepsilon$) are memoryless, but
in general need randomization \cite{crg-tcs07}.  
However, for concurrent safety games memoryless 
optimal strategies exist~\cite{dAM04}.
Thus, these strategies are mappings from states, to probability
distributions over moves. 

While complexity results are available for the solution of concurrent
reachability and safety games, practical algorithms for their
solution, that can provide both a value, and an estimated error, have
so far been lacking.
The question of whether the value of a concurrent reachability or
safety game is at least $p \in [0,1]$ can be decided in {\sc PSpace} via a
reduction to the theory of the real closed field \cite{EY06}. 
This yields a binary-search algorithm to approximate the value. 
This approach is theoretical, but complex due to the complex decision 
algorithms for the theory of reals. 

Thus far, the only practical approach to the solution of concurrent
safety and reachability games has been via value iteration, and
via strategy improvement for reachability games. 
In \cite{dAM04} it was shown how to construct a series of valuations
that approximates from below, and converges, to the value of a
reachability game; the same algorithm provides valuations converging
from above to the value of a safety game. 
In \cite{CdAH06}, it was shown how to construct a series of
strategies for reachability games that converge towards optimality.  
Neither scheme is guaranteed to terminate, not even strategy
improvement, since in general only $\varepsilon$-optimal strategies are
guaranteed to exist. 
Both of these approximation schemes lead to practical algorithms. 
The problem with both schemes, however, is that they provide only 
{\em lower\/} bounds for the value of reachability games, and only
{\em upper\/} bounds for the value of safety games. 
As no bounds are available for the speed of convergence of these
algorithms, the question of how to derive the matching bounds has so
far been open. 

In this paper, we present the first strategy improvement algorithm for
the solution of concurrent safety games. 
Given a safety goal for player~1, the algorithm computes a sequence of
memoryless, randomized strategies $\pi_1^0, \pi_1^1, \pi_1^2, \ldots$ for
player~1 that converge towards optimality. 
Albeit memoryless randomized optimal strategies exist for safety goals
\cite{dAM04}, the strategy improvement algorithm may not converge in
finitely many iterations: indeed, optimal strategies may require moves
to be played with irrational probabilities, while the strategies
produced by the algorithm play moves with probabilities that are
rational numbers.
The main significance of the algorithm is that it provides a converging
sequence of {\em lower\/} bounds for the value of a safety game, and
dually, of {\em upper\/} bounds for the value of a reachability game. 
To obtain such bounds, it suffices to compute the value $v_k(s)$
provided by $\pi_1^k$ at a state $s$, for $k > 0$. 
Once $\pi_1^k$ is fixed, the game is reduced to a Markov decision
process, and the value $v_k(s)$ of the safety game can be computed at
all $s$ e.g.\ via linear programming \cite{CY95,bda95}. 
Thus, together with the value or strategy improvement algorithms of
\cite{dAM04,CdAH06}, the algorithm presented in this paper
provides the first practical way of computing converging lower and
upper bounds for the values of concurrent reachability and safety
games. 
We also present a detailed analysis of termination criteria for turn-based
stochastic games, and obtain an improved upper bound for termination 
for turn-based stochastic games.

The strategy improvement algorithm for reachability games of
\cite{CdAH06} is based on locally improving the strategy on the
basis of the valuation it yields. 
This approach does not suffice for safety games: the sequence of
strategies obtained would yield increasing values to player~1, but
these value would not necessarily converge to the value of the game.
In this paper, we introduce a novel, and non-local, improvement step,
which augments the standard value-based improvement step.
The non-local step involves the analysis of an
appropriately-constructed turn-based game.  
As value iteration for safety games converges from above, while our
sequences of strategies yields values that converge from below, the
proof of convergence for our algorithm cannot be derived from a
connection with value iteration, as was the case for reachability
games. 
Thus, we developed new proof techniques to show both the monotonicity
of the strategy values produced by our algorithm, and to show
convergence to the value of the game.

\section{Definitions} 

\paragraph{Notation.} 
For a countable set~$A$, a {\em probability distribution\/} on $A$ is a 
function $\trans\!:A\to[0,1]$ such that $\sum_{a \in A} \trans(a) = 1$. 
We denote the set of probability distributions on $A$ by $\distr(A)$. 
Given a distribution $\trans \in \distr(A)$, we denote by 
$\supp(\trans) = \{x \in A \mid \trans(x) > 0\}$ the support set
of $\trans$.

\begin{definition}[Concurrent games]
A (two-player) {\em concurrent game structure\/} $\gamegraph = \langle S,
\moves, \mov_1, \mov_2, \trans \rangle$ consists of the following
components: 
\begin{itemize}

\item A finite state space $S$ and a finite set $\moves$ of moves or actions.

\item Two move assignments $\mov_1, \mov_2 \!: S\to 2^\moves
	\setminus \emptyset$.  For $i \in \{1,2\}$, assignment
	$\mov_i$ associates with each state $s \in S$ a nonempty
	set $\mov_i(s) \subseteq \moves$ of moves available to player $i$
	at state $s$.  

\item 
A probabilistic transition function 
$\trans: S \times \moves \times \moves \to \Distr(S)$ that gives the
probability $\trans(s, a_1, a_2)(t)$ of a transition from $s$ to
$t$ when player~1 chooses at state $s$ move $a_1$ and player~2 chooses 
move $a_2$, for all $s,t\in S$ and $a_1 \in \mov_1(s)$, $a_2 \in \mov_2(s)$.  
\end{itemize}
\end{definition}

\noindent
We denote by $|\trans|$ the size of transition 
function, i.e., $|\trans|=\sum_{s\in S,a \in \mov_1(s),b\in \mov_2(s),t\in S} 
|\trans(s,a,b)(t)|$, where $|\trans(s,a,b)(t)|$ is the number of bits 
required to specify the transition probability $\trans(s,a,b)(t)$.
We denote by $|\gamegraph|$ the size of the game graph, and $|G|=|\trans|+|S|$.
At every state $s\in S$, player~1 chooses a move $a_1\in\mov_1(s)$,
and simultaneously and independently player~2 chooses a move $a_2\in\mov_2(s)$.
The game then proceeds to the successor state $t$ with probability
$\trans(s,a_1,a_2)(t)$, for all $t \in S$.
A state $s$ is an \emph{absorbing state} if for all 
$a_1 \in \mov_1(s)$ and $a_2 \in \mov_2(s)$, we have 
$\trans(s, a_1,a_2)(s)=1$.
In other words, at an absorbing state $s$ for all choices of moves of the 
two players, the successor state is always $s$. 

\begin{definition}[Turn-based stochastic games]
A \emph{turn-based stochastic game graph} 
(\emph{$2\half$-player game graph})
$\gamegraph =\langle (S, E), (S_1,S_2,S_R),\trans\rangle$ 
consists of a finite directed graph $(S,E)$, a partition $(S_1$, $S_2$,
$S_R)$ of the finite set $S$ of states, and a probabilistic transition 
function $\trans$: $S_R \rightarrow \distr(S)$, where $\distr(S)$ denotes the 
set of probability distributions over the state space~$S$. 
The states in $S_1$ are the {\em player-$1$\/} states, where player~$1$
decides the successor state; the states in $S_2$ are the {\em 
player-$2$\/} states, where player~$2$ decides the successor state; 
and the states in $S_R$ are the {\em random or probabilistic\/} states, where
the successor state is chosen according to the probabilistic transition
function~$\trans$. 
We assume that for $s \in S_R$ and $t \in S$, we have $(s,t) \in E$ 
iff $\trans(s)(t) > 0$, and we often write $\trans(s,t)$ for $\trans(s)(t)$. 
For technical convenience we assume that every state in the graph 
$(S,E)$ has at least one outgoing edge.
For a state $s\in S$, we write $E(s)$ to denote the set 
$\set{t \in S \mid (s,t) \in E}$ of possible successors.
We denote by $|\trans|$ the size of the transition 
function, i.e., $|\trans|=\sum_{s\in S_R,t\in S} 
|\trans(s)(t)|$, where $|\trans(s)(t)|$ is the number of bits 
required to specify the transition probability $\trans(s)(t)$.
We denote by $|\gamegraph|$ the size of the game graph, and 
$|G|=|\trans|+|S|+|E|$.
\end{definition}

\paragraph{Plays.}
A \emph{play} $\pat$ of $\gamegraph$ is an infinite sequence
$\pat = \langle s_0,s_1,s_2,\ldots \rangle $ of states in $S$ such that 
for all $k\ge 0$, there are moves $a^k_1 \in \mov_1(s_k)$ and 
$a^k_2 \in \mov_2(s_k)$ with $\trans(s_k,a^k_1,a^k_2)(s_{k+1}) >0$.
We denote by $\Paths$ the set of all plays, and by $\Paths_s$ the set of all 
plays $\pat=\seqs$ such that $s_0=s$, that is, the set of plays starting 
from state~$s$. 

\paragraph{Selectors and strategies.}
A \emph{selector} $\xi$ for player $i \in \set{1,2}$ is a function
$\xi : S \to \Distr(\moves)$ such that for all states $s \in S$ and
moves $a \in \moves$, if $\xi(s)(a) > 0$, then $a \in \mov_i(s)$.
A selector $\xi$ for player $i$ at a state $s$ is a distribution 
over moves such that if $\xi(s)(a)>0$, then $a \in \mov_i(s)$.
We denote by $\Sel_i$ the set of all selectors for player~$i \in \set{1,2}$,
and similarly, we denote by $\Sel_i(s)$ the set of all selectors for
player~$i$ at a state $s$.
The selector $\xi$ is {\em pure\/} if for every state $s \in S$, there is 
a move $a \in \moves$ such that $\xi(s)(a) = 1$.
A \emph{strategy} for player $i\in\set{1,2}$ is a function 
$\stra: S^+ \to \Distr(\moves)$ that
associates with every finite, nonempty sequence of states,
representing the history of the play so far, a selector for player~$i$; 
that is, for all $w \in S^*$ and $s \in S$, we have 
$\supp(\stra(w \cdot s)) \subs \mov_i(s)$.
The strategy $\stra$ is {\em pure\/} 
if it always chooses a pure selector;
that is, for all $w \in S^+$, there is a move $a \in \moves$ such that 
$\stra(w)(a)=1$.
A \emph{memoryless} strategy is independent of the history of the play and
depends only on the current state. 
Memoryless strategies correspond to selectors; we write
$\overline{\xi}$ for the memoryless strategy consisting in playing
forever the selector $\xi$. 
A strategy is \emph{pure memoryless} if it is both pure and
memoryless.
In a turn-based stochastic game, a strategy for player~1 is a function
$\stra_1:S^* \cdot S_1 \to \Distr(S)$, such that for all $w \in S^*$
and for all $s \in S_1$ we have $\supp(\stra_1(w\cdot s)) \subs E(s)$.
Memoryless strategies and pure memoryless strategies are obtained 
as the restriction of strategies as in the case of concurrent game 
graphs.
The family of strategies for player~2 are defined analogously.
We denote by $\Stra_1$ and $\Stra_2$ the sets of all strategies for 
player $1$ and player $2$, respectively.
We denote by $\Stra_i^M$ and $\Stra_i^\PM$ the sets of memoryless strategies 
and pure memoryless strategies for player~$i$, respectively. 

\paragraph{Destinations of moves and selectors.}
For all states $s \in S$ and moves $a_1 \in \mov_1(s)$ and $a_2 \in \mov_2(s)$,
we indicate by $\dest(s,a_1,a_2) = \supp(\trans(s,a_1,a_2))$
the set of possible successors of $s$ when the moves $a_1$ and $a_2$ are 
chosen.
Given a state $s$, and selectors $\xi_1$ and $\xi_2$ for the two players, 
we denote by 
\[
  \dest(s,\xi_1,\xi_2) = \bigcup_{\begin{array}{c}
      \scriptstyle a_1 \in \supp(\xi_1(s)),\\ 
      \scriptstyle a_2 \in \supp(\xi_2(s)) \end{array}} \dest(s,a_1,a_2)
\] 
the set of possible successors of $s$ with respect to the 
selectors $\xi_1$ and $\xi_2$. 

Once a starting state $s$ and strategies $\stra_1$ and $\stra_2$
for the two players are fixed, the game is reduced to an
ordinary stochastic  process.
Hence, the probabilities of events are uniquely defined, where an {\em
event\/} $\cala\subseteq\Paths_s$ is a measurable set of
plays.
For an event $\cala\subseteq\Paths_s$, we denote by
$\Prb_s^{\stra_1,\stra_2}(\cala)$ the probability that a play belongs to
$\cala$ when the game starts from $s$ and the players follows the
strategies $\stra_1$ and~$\stra_2$.
Similarly, for a measurable function $f: \Paths_s \to \reals$, 
we denote by $\E_s^{\stra_1,\stra_2}(f)$ the expected value of $f$ when
the game starts from $s$ and the players follow the strategies
$\stra_1$ and~$\stra_2$.
For $i \geq 0$, we denote by $\randpath_i: \Paths
\to S$ the random  variable denoting the $i$-th state along a
play. 

\paragraph{Valuations.}
A {\em valuation\/} is a mapping $v: S \to [0,1]$ associating a real
number $v(s) \in [0,1]$ with each state $s$. 
Given two valuations $v, w: S \to \reals$, we write $v \leq w$ when
$v(s) \leq w(s)$ for all states $s \in S$. 
For an event $\cala$, we denote by 
$\Prb^{\stra_1,\stra_2}(\cala)$ the valuation $S \to [0,1]$
defined for all states $s \in S$ 
by $\bigl(\Prb^{\stra_1,\stra_2}(\cala)\bigr)(s) =
\Prb_s^{\stra_1,\stra_2}(\cala)$. 
Similarly, for a measurable function $f: \Omega_s \to [0,1]$,
we denote by $\E^{\stra_1,\stra_2}(f)$ the valuation $S \to [0,1]$
defined for all $s \in S$ by $\bigl(\E^{\stra_1,\stra_2}(f)\bigr)(s) =
\E_s^{\stra_1,\stra_2}(f)$. 

\paragraph{Reachability and safety objectives.}
Given a set $F \subs S$ of \emph{safe} states, the objective of a safety 
game consists in never leaving $F$.
Therefore, we define the set of winning plays as the set 
$\Safe(F)=\set{\seq{s_0,s_1, s_2,\ldots} \in \Paths \mid s_k \in F 
\mbox{ for all } k \geq 0}$.
Given a subset $T \subs S$ of \emph{target} states, the objective of a
reachability game consists in reaching $T$. 
Correspondingly, the set winning plays is 
$\Reach(T) = \set{\seq{s_0, s_1, s_2,\ldots} \in \Paths \mid s_k
\in T \mbox{ for some }k \ge 0}$ of plays that visit $T$.
For all $F \subs S$ and $T \subs S$, the sets $\Safe(F)$ and $\Reach(T)$ 
is measurable.
An objective in general is a measurable set, and in this paper we would
consider only reachability and safety objectives.
For an objective $\Phi$, the probability of satisfying 
$\Phi$ from a state $s \in S$ under strategies $\stra_1$ and $\stra_2$ 
for players~1 and~2, respectively,
is $\Prb_s^{\stra_1,\stra_2}(\Phi)$.
We define the \emph{value} for player~1 of game with objective $\Phi$ 
from the state $s \in S$ as 
\[
  \va(\Phi)(s) =
  \sup_{\stra_1\in\Stra_1}\inf_{\stra_2\in\Stra_2}
  \Prb_s^{\stra_1,\stra_2}(\Phi); 
\]
i.e., the value is the maximal probability with which player~1 can 
guarantee the satisfaction of $\Phi$ against all player~2 strategies.
Given a player-1 strategy $\stra_1$, we use the notation
\[
\winval{1}^{\stra_1}(\Phi)(s)
= \inf_{\stra_2 \in \Stra_2} \Prb_s^{\stra_1,\stra_2}(\Phi).
\]
A strategy $\stra_1$ for player~1 is {\em optimal\/} for an
objective $\Phi$ if for all 
states $s \in S$, we have 
\[
\winval{1}^{\stra_1}(\Phi)(s)= \va(\Phi)(s). 
\]
For $\vare > 0$, a strategy $\stra_1$ for player~1 is 
{\em $\vare$-optimal\/} if for all states $s \in S$, we have 
\[
\winval{1}^{\stra_1}(\Phi)(s)
  \geq \va(\Phi)(s) - \vare. 
\]
The notion of values and optimal strategies for player~2 are
defined analogously.
Reachability and safety objectives are dual, i.e., we have
$\Reach(T)=\Paths \setminus \Safe(S\setminus T)$.  
The quantitative determinacy result of~\cite{Martin98} ensures that for all
states $s \in S$, we have
\[  
\va(\Safe(F))(s) + \vb(\Reach(S\setminus F))(s)  = 1. 
\]

\begin{theorem}[Memoryless determinacy]\label{thrm:memory-determinacy}
For all concurrent game graphs $G$, for all $F,T \subs S$, such
that $F=S \setminus T$, the following assertions hold.
\begin{enumerate}
\item \cite{FV97} Memoryless optimal strategies exist
for safety objectives $\Safe(F)$.
\item \cite{CdAH06,EY06} For all $\vare>0$, memoryless
$\vare$-optimal strategies exist for reachability
objectives $\Reach(T)$.

\item \cite{Con92} If $G$ is a turn-based stochastic game
graph, then pure memoryless optimal strategies exist
for reachability objectives $\Reach(T)$ and safety 
objectives $\Safe(F)$. 
\end{enumerate}
\end{theorem}

\section{Markov Decision Processes}\label{sec:mdp}

\noindent
To develop our arguments, we need some facts about one-player versions of
concurrent stochastic games, known as \emph{Markov decision processes}
(MDPs) \cite{Derman,Bertsekas95}. 
For $i \in \set{1,2}$, a \emph{player-$i$ MDP} 
(for short, $i$-MDP) is a concurrent game where, for all
states $s \in S$, we have $|\mov_{3-i}(s)|=1$.
Given a concurrent game $G$, if we fix a memoryless strategy
corresponding to selector $\xi_1$ for player~1, the game
is equivalent to a 2-MDP $G_{\xi_1}$ with the transition function 
\[
\trans_{\xi_1}(s,a_2)(t) = \sum_{a_1 \in \mov_1(s)} 
\trans(s,a_1,a_2)(t) \cdot \xi_1(s)(a_1),
\]
for all $s \in S$ and $a_2 \in \mov_2(s)$. 
Similarly, if we fix selectors $\xi_1$ and $\xi_2$ for both players in a
concurrent game $G$, we obtain a Markov chain, which we denote by
$G_{\xi_1,\xi_2}$. 

\paragraph{End components.}

In an MDP, the sets of states that play an equivalent role to the 
closed recurrent classes of Markov chains \cite{Kemeny} are called
``end components'' \cite{CY95,luca-thesis}. 

\begin{definition}[End components]
An {\em end component\/} of an $i$-MDP $G$, for $i\in\set{1,2}$, is a 
subset $C \subs S$ of the states such that there is a selector $\xi$ 
for player~$i$ so that 
$C$ is a closed recurrent class of the Markov chain $G_\xi$. 
\end{definition}

\noindent
It is not difficult to see that an equivalent characterization of an
end component $C$ is the following. 
For each state $s \in C$, there is a subset $M_i(s) \subs \mov_i(s)$
of moves such that:
\begin{enumerate}
\item {\em (closed)} if a move in $M_i(s)$ is chosen 
by player $i$ at state $s$, then all successor states that are obtained 
with nonzero probability lie in $C$; and

\item {\em (recurrent)} the graph $(C,E)$, where $E$
consists of the transitions that occur with nonzero probability when
moves in $M_i(\cdot)$ are chosen by player $i$, is strongly connected.

\end{enumerate}
Given a play $\pat\in\Paths$, 
we denote by $\infi(\pat)$ the set of states that
occurs infinitely often along $\pat$.  
Given a set $\calf \subs 2^S$ of subsets of states, we denote by 
$\infi(\calf)$ the event $\set{\pat \mid \infi(\pat) \in \calf}$.
The following theorem states that in a 2-MDP, for every strategy of
player~2, the set of states that are visited infinitely often is,
with probability~1, an end component. 
Corollary~\ref{coro:prob1} follows easily from Theorem~\ref{theo-ec}.

\begin{theorem}{\rm\cite{luca-thesis}} \label{theo-ec}
For a player-1 selector $\xi_1$, 
let $\calc$ be the set of end components of a 2-MDP $G_{\xi_1}$.
For all player-2 strategies $\stra_2$ and all states $s \in S$, we have
$\Prb_s^{\ov{\xi}_1,\stra_2}(\infi(\calc)) = 1$. 
\end{theorem}

\begin{corollary}\label{coro:prob1}
For a player-1 selector $\xi_1$, 
let $\calc$ be the set of end components of a 2-MDP $G_{\xi_1}$, and
let $Z= \bigcup_{C \in \calc} C$ be the set of states of all
end components.
For all player-2 strategies $\stra_2$ and all states $s \in S$, we have
$\Prb_s^{\ov{\xi}_1,\stra_2}(\Reach(Z)) = 1$. 
\end{corollary}

\paragraph{MDPs with reachability objectives.}\label{subsec:mdpreach}

Given a 2-MDP with a reachability objective $\Reach(T)$ for player~2,
where $T\subseteq S$, the values can be obtained as the solution of a 
linear program~\cite{FV97}.
The linear program has a variable $x(s)$ for all states $s \in S$, and the
objective function and the constraints are as follows:
\[
\min \ \displaystyle \sum_{s \in S} x(s) \quad \text{subject to } \\
\]
\vspace*{-2ex}
\begin{align*}
 x(s) \geq \displaystyle \sum_{t \in S} x(t) \cdot \trans(s,a_2) (t) 
	& \text{\ \ for all \ } s\in S \text{\ and\ } 
        a_2 \in \mov_2(s) \\
 x(s) = 1 & \text{\ \ for all \ } s \in T \\
 0 \leq x(s) \leq 1 & \text{\ \ for all \ } s \in S
\end{align*}
The correctness of the above linear program to compute the values follows 
from~\cite{Derman,FV97}.

\section{Strategy Improvement for Safety Games}

\noindent
In this section we present a strategy improvement 
algorithm for concurrent games with safety objectives. 
The algorithm will produce a sequence of selectors 
$\gamma_0, \gamma_1, \gamma_2, \ldots$ for player 1, such that: 
\begin{enumerate}
  \item \label{l-improve-1} 
  for all $i \geq 0$, we have 
  $\vas{\overline{\gamma}_i}(\Safe(F)) \leq \vas{\overline{\gamma}_{i+1}}(\Safe(F))$;

  \item \label{l-improve-3} 
  if there is $i \geq 0$ such that $\gamma_i = \gamma_{i+1}$, 
  then $\vas{\overline{\gamma}_i}(\Safe(F)) = \va(\Safe(F))$; and 

\item \label{l-improve-2} 
  $\lim_{i \rightarrow \infty} \vas{\overline{\gamma}_i}(\Safe(F)) = \va(\Safe(F))$. 
\end{enumerate}
Condition~\ref{l-improve-1} guarantees that the algorithm computes a
sequence of monotonically improving selectors. 
Condition~\ref{l-improve-3} guarantees that if a selector cannot be
improved, then it is optimal. 
Condition~\ref{l-improve-2} guarantees that the value guaranteed by
the selectors converges to the value of the game, or equivalently,
that for all $\vare > 0$, there is a number $i$ of iterations
such that the memoryless player-1 strategy $\ov{\gamma}_i$ is 
$\vare$-optimal. 
Note that for concurrent safety games, there may be no $i \geq
0$ such that $\gamma_i = \gamma_{i+1}$, that is, the algorithm may fail to
generate an optimal selector. 
This is because there are concurrent safety games such that the
values are irrational~\cite{dAM04}.
We start with a few notations

\medskip\noindent{\bf The $\Pre$ operator and optimal selectors.}
Given a valuation $v$, and two selectors 
$\xi_1 \in \Sel_1$ and $\xi_2 \in \Sel_2$, 
we define the valuations 
$\Pre_{\xi_1,\xi_2}(v)$, $\Pre_{1:\xi_1}(v)$, and $\Pre_1(v)$ 
as follows, for all states $s \in S$: 
\begin{multline*} 
  \Pre_{\xi_1,\xi_2}(v)(s) 
  = \sum_{a,b \in \moves} \, \sum_{t \in S} v(t) \cdot \trans(s,a,b)(t) \cdot
    \xi_1(s)(a) \cdot \xi_2(s)(b)
  \\ \quad
  \Pre_{1:\xi_1}(v)(s) 
   =  \inf_{\xi_2 \in \Sel_2} \, \Pre_{\xi_1 ,\xi_2}(v)(s) \hfill
  \\ \quad
  \Pre_1(v)(s) 
   = \sup_{\xi_1 \in \Sel_1} \, \inf_{\xi_2 \in \Sel_2} \, 
  \Pre_{\xi_1 ,\xi_2}(v)(s) \hfill
\end{multline*}
Intuitively, $\Pre_1(v)(s)$ is the greatest expectation of $v$ that
player~1 can guarantee at a successor state of $s$. 
Also note that given a valuation $v$, the computation of $\Pre_1(v)$ 
reduces to the solution of a zero-sum one-shot matrix game, and can be 
solved by linear programming.
Similarly, 
$\Pre_{1:\xi_1}(v)(s)$ is the greatest expectation of $v$ that
player~1 can guarantee at a successor state of $s$ by playing the 
selector $\xi_1$. 
Note that all of these operators on valuations are monotonic: 
for two valuations $v, w$, if $v \leq w$,
then for all selectors $\xi_1 \in \Sel_1$ and $\xi_2 \in \Sel_2$, we have 
$\Pre_{\xi_1,\xi_2}(v) \leq \Pre_{\xi_1,\xi_2}(w)$, 
$\Pre_{1:\xi_1}(v) \leq \Pre_{1:\xi_1}(w)$, 
and $\Pre_1(v) \leq \Pre_1(w)$. 
Given a valuation $v$ and a state $s$, we define by
\[
\OptSel(v,s) =\set{\xi_1 \in \Sel_1(s) \mid \Pre_{1:\xi_1}(v)(s) =\Pre_1(v)(s)}
\]
the set of optimal selectors for $v$ at state $s$.
For an optimal selector $\xi_1 \in \OptSel(v,s)$, we define the set
of counter-optimal actions as follows:
\[
\CountOpt(v,s,\xi_1) =\set{ b \in \mov_2(s) \mid 
\Pre_{\xi_1,b}(v)(s) =\Pre_1(v)(s)}.
\]
Observe that for $\xi_1 \in \OptSel(v,s)$, for all $b \in 
\mov_2(s) \setminus \CountOpt(v,s,\xi_1)$ we have 
$\Pre_{\xi_1,b}(v)(s) > \Pre_1(v)(s)$.
We define the set of optimal selector support and the counter-optimal 
action set as follows:
\[
\begin{array}{rcl}
\OptSelCount(v,s) & = & \set{(A,B) \subs \mov_1(s) \times \mov_2(s) \mid
\exists \xi_1 \in \Sel_1(s). \ \xi_1 \in \OptSel(v,s) \\
 & \land & 
\supp(\xi_1)=A \ \land \ \CountOpt(v,s,\xi_1)=B
};
\end{array}
\]
i.e., it consists of pairs $(A,B)$ of actions of player~1 and player~2,
such that there is an optimal selector $\xi_1$ with support $A$,
and $B$ is the set of counter-optimal actions to $\xi_1$.

\medskip\noindent{\bf Turn-based reduction.} Given a concurrent 
game $G=\langle S,\moves,\mov_1,\mov_2, \trans \rangle $ and 
a valuation $v$ we construct a turn-based stochastic game
$\ov{G}_v=\langle (\ov{S},\ov{E}), (\ov{S}_1,\ov{S}_2,\ov{S}_R),\ov{\trans}
\rangle$ as follows:
\begin{enumerate}
\item The set of states is as follows:
\[
\begin{array}{rcl}
\ov{S}& = & S \cup \set{(s,A,B) \mid s\in S, \ (A,B) \in \OptSelCount(v,s)} \\
	&\cup & \set{(s,A,b) \mid s \in S, \ (A,B) \in \OptSelCount(v,s), \ b \in B}.
\end{array}
\]

\item The state space partition is as follows: 
$\ov{S}_1=S$; $\ov{S}_2=\set{(s,A,B) \mid s \in S, (A,B) \in \OptSelCount(v,s)}$;
and $\ov{S}_R=\ov{S} \setminus (\ov{S}_1 \cup \ov{S}_2).$

\item The set of edges is as follows:
\[ 
\begin{array}{rcl}
\ov{E} & = & \set{(s,(s,A,B)) \mid s \in S, (A,B) \in \OptSelCount(v,s)} \\
	& \cup & \set{((s,A,B),(s,A,b)) \mid b \in B} 
	\cup \set{((s,A,b),t) \mid \displaystyle t \in \bigcup_{a \in A} \dest(s,a,b)}.
\end{array}
\]

\item The transition function $\ov{\trans}$ for all states in $\ov{S}_R$ 
is uniform over its successors.
\end{enumerate}
Intuitively, the reduction is as follows.
Given the valuation $v$, state $s$ is a player~1 state where
player~1 can select a pair $(A,B)$ (and move to
state $(s,A,B)$) with $A \subs \mov_1(s)$ 
and $B \subs \mov_2(s)$ such that there is an optimal 
selector $\xi_1$ with support exactly $A$ and the set of
counter-optimal actions to $\xi_1$ is the set $B$.
From a player~2 state $(s,A,B)$, player~2 can choose any action
$b$ from the set $B$, and move to state $(s,A,b)$.
A state $(s,A,b)$ is a probabilistic state where all the 
states in $\bigcup_{a\in A} \dest(s,a,b)$ are chosen 
uniformly at random.
Given a set $F \subseteq S$ we denote by $\ov{F}= F \cup 
\set{(s,A,B) \in\ov{S} \mid s \in F} \cup 
\set{(s,A,b) \in\ov{S} \mid s \in F}$.
We refer to the above reduction as $\TB$, i.e., 
$(\ov{G}_v,\ov{F})=\TB(G,v,F)$.

\medskip\noindent{\bf Value-class of a valuation.}
Given a valuation $v$ and a real $0\leq r \leq 1$, 
the \emph{value-class} $U_r(v)$ of value $r$ is the set of 
states with valuation $r$, i.e., 
$U_r(v)=\set{s \in S \mid v(s)=r}$

\subsection{The strategy improvement algorithm} 

\paragraph{Ordering of strategies.}
Let $G$ be a concurrent game and $F$ be the set of safe states.
Let $T =S \setminus F$.
Given a concurrent game graph $G$ with a safety objective $\Safe(F)$,
the set of \emph{almost-sure winning} states is the set of states $s$ such that
the value at $s$ is~$1$, i.e., $W_1=\set{s\in S \mid \va(\Safe(F))=1}$
is the set of almost-sure winning states.
An optimal strategy from $W_1$ is referred as an almost-sure winning 
strategy.
The set $W_1$ and an almost-sure winning strategy can be computed in 
linear time by the algorithm given in~\cite{dAH00}.
We assume without loss of generality that all states in $W_1 \union T$
are absorbing. 
We define a preorder $\prec$ on the strategies for player 1 as follows:
given two player 1 strategies $\stra_1$ and $\stra_1'$, let
$\stra_1 \prec \stra_1'$ if the following two conditions hold:
(i)~$\vas{\stra_1}(\Safe(F)) \leq \vas{\stra_1'}(\Safe(F))$; and 
(ii)~$\vas{\stra_1}(\Safe(F))(s) < \vas{\stra_1'}(\Safe(F))(s)$ 
for some state $s\in S$.
Furthermore, we write 
$\stra_1 \preceq \stra_1'$ if either $\stra_1 \prec \stra_1'$ or 
$\stra_1 = \stra_1'$.
We first present an example that shows the improvements 
based only on $\Pre_1$ operators are not sufficient for 
safety games, even on turn-based games and then present our algorithm.

\begin{example}\label{examp:conc-safety}
Consider the turn-based stochastic game shown in Fig~\ref{fig:example-tbs}, where
the $\Box$ states are player~1 states, the $\Diamond$ states are 
player~2 states, and $\bigcirc$ states are random states with probabilities
labeled on edges.
The safety goal is to avoid the state $s_6$. 
Consider a memoryless strategy $\stra_1$ for player~1 that chooses the successor $s_0 \to
s_2$, and the counter-strategy $\stra_2$ for player~2 chooses $s_1 \to s_0$.
Given the strategies $\stra_1$ and $\stra_2$, the value at 
$s_0,s_1$ and $s_2$ is $1/3$, and since all
successors of $s_0$ have value $1/3$, the value cannot be improved by $\Pre_1$.
However, note that if player~2 is restricted to choose only value optimal 
selectors for the value $1/3$, then player~1 can switch to the strategy 
$s_0\to s_2$ and ensure that the game stays in the value class $1/3$ 
with probability~1.
Hence switching to $s_0 \to s_2$ would force player~2 to select a 
counter-strategy that switches to the strategy $s_1 \to s_3$, and
thus player~1 can get a value $2/3$.
\begin{figure}[t]
   \begin{center}
      \input{example1.eepic}
   \end{center}
  \caption{A turn-based stochastic safety game.}
  \label{fig:example-tbs}
\end{figure}
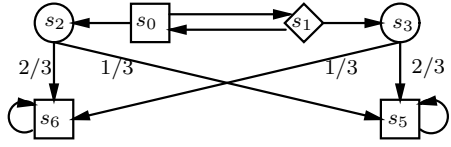
\qed
\end{example}

\paragraph{Informal description of Algorithm~\ref{algorithm:strategy-improve-safe}.}
We now present the strategy improvement algorithm 
(Algorithm~\ref{algorithm:strategy-improve-safe}) for computing the values for 
all states in $S \setminus W_1$.
The algorithm iteratively improves player-1 strategies according to the 
preorder $\prec$. 
The algorithm starts with the random selector 
$\gamma_0=\overline{\xi}_1^\unif$ that plays at all states all actions
uniformly at random.
At iteration $i+1$, the algorithm considers the memoryless player-1 strategy
$\overline{\gamma}_i$ and computes the value $\vas{\overline{\gamma}_i}(\Safe(F))$.
Observe that since $\overline{\gamma}_i$ is a memoryless strategy, the
computation of $\vas{\overline{\gamma}_i}(\Safe(F))$ involves solving the 2-MDP 
$G_{{\gamma}_i}$.
The valuation $\vas{\overline{\gamma}_i}(\Safe(F))$ is named $v_i$.
For all states $s$ such that $\Pre_1(v_i)(s) > v_i(s)$, 
the memoryless strategy at $s$ is modified to a selector 
that is value-optimal for $v_i$. 
The algorithm then proceeds to the next iteration.
If $\Pre_1(v_i) = v_i$, then the algorithm constructs the 
game $(\ov{G}_{v_i},\ov{F})=\TB(G,v_i,F)$, and computes
$\ov{A}_i$ as the set of almost-sure winning states in $\ov{G}_{v_i}$
for the objective $\Safe(\ov{F})$.
Let $U=(\ov{A}_i \cap S) \setminus W_1$.
If $U$ is non-empty, then a selector $\gamma_{i+1}$ is obtained at $U$ 
from an pure memoryless optimal strategy (i.e.,
an almost-sure winning strategy) in $\ov{G}_{v_i}$, and
the algorithm proceeds to iteration $i+1$.
If $\Pre_1(v_i)=v_i$ and $U$ is empty, then the algorithm stops and returns 
the memoryless strategy $\overline{\gamma}_i$ for player~1.
Unlike strategy improvement algorithms for turn-based games (see
\cite{Con93} for a survey),
Algorithm~\ref{algorithm:strategy-improve-safe} is not guaranteed to
terminate, because the value of a safety game may not be
rational. 

\begin{algorithm*}[t]
\caption{Safety Strategy-Improvement Algorithm}
\label{algorithm:strategy-improve-safe}
{
\begin{tabbing}
aaa \= aaa \= aaa \= aaa \= aaa \= aaa \= aaa \= aaa \kill
\\
\> {\bf Input:} a concurrent game structure $G$ with safe set $F$. \\
\>   {\bf Output:} a strategy $\overline{\gamma}$ for player~1. \\ 

\> 0. Compute $W_1=\set{s \in S \mid \va(\Safe(F))(s)=1}$. \\
\> 1. Let $\gamma_0=\xi_1^\unif$ and $i=0$. \\
\> 2. Compute $v_0 = \vas{\overline{\gamma}_0}(\Safe(F))$. \\

\> 3. {\bf do \{ } \\ 
\>\> 3.1. Let $I= \set{s \in S \setminus (W_1 \cup T) \mid \Pre_1(v_i)(s) > v_i(s)}$. \\
\>\> 3.2 {\bf if} $I \neq \emptyset$, {\bf then} \\
\>\>\> 3.2.1 Let $\xi_1$ be a player-1 
        selector such that for all states $s \in I$, \\
\>\>\>\>	we have $\Pre_{1:\xi_1}(v_i)(s) =\Pre_1(v_i)(s) > v_i(s)$.\\
\>\>\> 3.2.2 The player-1 selector $\gamma_{i+1}$ is defined
          as follows: for each state $t\in S$, let\\
\>\>\>\> $	\displaystyle 
	\gamma_{i+1}(t)=
	\begin{cases}
	\gamma_i(t) & \text{\ if \ }s\not\in I;\\
	\xi_1(s) & \text{\ if\ }s\in I.
	\end{cases}$
	\\ 
\>\> 3.3 {\bf else}  \\
\>\>\> 3.3.1 let$(\ov{G}_{v_i},\ov{F})=\TB(G,v_i,F)$ \\
\>\>\> 3.3.2 let $\ov{A}_i$ be the set of almost-sure winning states in $\ov{G}_{v_i}$
	for $\Safe(\ov{F})$ and \\
\>\>\>\> $\ov{\stra}_1$ be a pure memoryless almost-sure winning strategy from the set $\ov{A}_i$.\\
\>\>\> 3.3.3 {\bf if} ($(\ov{A}_i \cap S) \setminus W_1 \neq \emptyset$) \\
\>\>\>\> 3.3.3.1 let $U= (\ov{A}_i \cap S)\setminus W_1$ \\
\>\>\>\> 3.3.3.2 The player-1 selector $\gamma_{i+1}$ is defined
          as follows: for $t\in S$, let\\
\>\>\>\> $	\displaystyle 
	\gamma_{i+1}(t)=
	\begin{cases}
	\gamma_i(t) & \text{\ if \ }s\not\in U;\\
	\xi_1(s) & \text{\ if\ }s\in U, \xi_1(s) \in \OptSel(v_i,s), \\
	&\ \ \ov{\stra}_1(s)=(s,A,B), B=\OptSelCount(s,v,\xi_1).
	\end{cases}$
	\\
\>\> 3.4. Compute $v_{i+1} =\vas{\overline{\gamma}_{i+1}}(\Safe(F))$. \\
\>\> 3.5. Let $i=i+1$. \\
\> {\bf \} until } $I=\emptyset$ and $(\ov{A}_{i-1} \cap S) \setminus W_1=\emptyset$. \\
\> 4. {\bf return} $\overline{\gamma}_{i}$.  
\end{tabbing}
}
\end{algorithm*}

\begin{lemma}\label{lemm:stra-improve-safe1}
Let $\gamma_i$ and $\gamma_{i+1}$ be the player-1 selectors obtained at 
iterations $i$ and $i+1$ of Algorithm~\ref{algorithm:strategy-improve-safe}.
Let $I=\set{s \in S \setminus (W_1 \cup T) \mid \Pre_1(v_i)(s) > v_i(s)}$. 
Let $v_i=\vas{\overline{\gamma}_i}(\Safe(F))$ and 
$v_{i+1}=\vas{\overline{\gamma}_{i+1}}(\Safe(F))$.
Then $v_{i+1}(s)  \geq  \Pre_1(v_i)(s)$ for all states $s\in S$;
and therefore
$v_{i+1}(s) \geq v_i(s)$ for all states $s\in S$, 
and $v_{i+1}(s) > v_i(s)$ for all states $s\in I$.
\end{lemma}
\begin{proof}
Consider the valuations $v_i$ and $v_{i+1}$ obtained at iterations $i$ and 
$i+1$, respectively, and let $w_i$ be the valuation defined by 
$w_i(s) = 1 - v_i(s)$ for all states $s \in S$. 
The counter-optimal strategy for player~2 to minimize
$v_{i+1}$ is obtained by maximizing the probability to reach $T$.
Let 
\[
  w_{i+1}(s) =
  \begin{cases} 
    w_i(s) & \text{\ if\ }s \in S \setminus I; \\
    1-\Pre_1(v_i)(s) < w_i(s) &\text{\ if\ }s \in I. 
  \end{cases}
\]
In other words, $w_{i+1} = 1 - \Pre_1(v_i)$, and we also have 
$w_{i+1} \leq w_i$.
We now show that $w_{i+1}$ is a feasible solution to the 
linear program for MDPs with the objective $\Reach(T)$, as
described in Section~\ref{sec:mdp}.
Since $v_i =\vas{\overline{\gamma}_i}(\Safe(F))$, it follows that
for all states $s\in S$ and all moves $a_2 \in \mov_2(s)$, we have
\[
  w_i(s) \geq \sum_{t \in S} w_i(t) \cdot \trans_{\gamma_i}(s,a_2).  
\]
For all states $s \in S \setminus I$, 
we have $\gamma_i(s)=\gamma_{i+1}(s)$ and
$w_{i+1}(s) = w_i(s)$,
and since $w_{i+1} \leq w_i$, it follows that for all 
states $s \in S \setminus I$ and all moves $a_2 \in \mov_2(s)$, we have
\[
  w_{i+1}(s)=w_i(s) \geq \sum_{t \in S} w_{i+1}(t) \cdot \trans_{\gamma_{i+1}}(s,a_2) 
\qquad \text{( for  $s\in S\setm I$)}.
\]

Since for $s \in I$ the selector $\gamma_{i+1}(s)$ is obtained as an
optimal selector for $\Pre_1(v_i)(s)$, it follows that 
for all states $s\in I$ and all moves $a_2 \in \mov_2(s)$, we have
\[
\Pre_{\gamma_{i+1},a_2}(v_i)(s) \geq \Pre_1(v_i)(s);
\]
in other words, $1-\Pre_1(v_i)(s) \geq 1-\Pre_{\gamma_{i+1},a_2}(v_i)(s)$.
Hence for all states $s\in I$ and all moves $a_2 \in \mov_2(s)$, we have
\[
  w_{i+1}(s) \geq \sum_{t \in S} w_{i}(t) \cdot \trans_{\gamma_{i+1}}(s,a_2). 
\]
Since $w_{i+1} \leq w_i$,  
for all states $s\in I$ and all moves $a_2 \in \mov_2(s)$, we have
\[
  w_{i+1}(s) \geq \sum_{t \in S} w_{i+1}(t) \cdot \trans_{\gamma_{i+1}}(s,a_2) 
\qquad \text{( for  $s\in I$)}.
\] 
Hence it follows that $w_{i+1}$ is a feasible solution to the 
linear program for MDPs with reachability objectives.
Since the reachability valuation 
for player~2 for $\Reach(T)$ is the least
solution (observe that the objective function of the linear program is a 
minimizing function), it follows that 
$v_{i+1} \geq 1-w_{i+1} = \Pre_1(v_i)$.
Thus we obtain
$v_{i+1}(s) \geq v_i(s)$ for all states $s \in S$, and 
$v_{i+1}(s) > v_i(s)$ for all states $s \in I$.
\qed
\end{proof}

Recall that by Example~\ref{examp:conc-safety} it follows that 
improvement by only step~3.2 is not sufficient to guarantee
convergence to optimal values.
We now present a lemma about the turn-based reduction,
and then show that step 3.3 also leads to an improvement.
Finally, in Theorem~\ref{thrm:safe-termination} we show that if
improvements by step 3.2 and step 3.3 are not possible, then 
the optimal value and an optimal strategy is obtained.

\begin{lemma}\label{lemm:stra-improve-safetb}
Let $G$ be a concurrent game with a set $F$ of safe states.
Let $v$ be a valuation and 
consider $(\ov{G}_v,\ov{F})=\TB(G,v,F)$.
Let $\ov{A}$ be the set of almost-sure winning states in $\ov{G}_v$
for the objective $\Safe(\ov{F})$, and let $\ov{\stra}_1$ be a
pure memoryless almost-sure winning strategy from $\ov{A}$ in
$\ov{G}_v$.
Consider a memoryless strategy $\stra_1$ in $G$ for 
states in $\ov{A}\cap S$ as follows:
if $\ov{\stra}_1(s)=(s,A,B)$, then $\stra_1(s) \in 
\OptSel(v,s)$ such that $\supp(\stra_1(s))=A$ and $\OptSelCount(v,s,\stra_1(s))=B$.
Consider a pure memoryless strategy $\stra_2$ 
for player~2.
If for all states $s \in \ov{A} \cap S$, we have 
$\stra_2(s) \in \OptSelCount(v,s,\stra_1(s))$, then 
for all $s \in \ov{A} \cap S$, we have 
$\Prb_s^{\stra_1,\stra_2}(\Safe(F))=1$. 
\end{lemma}
\begin{proof}
We analyze the Markov chain arising after the player fixes
the memoryless strategies $\stra_1$ and $\stra_2$.
Given the strategy $\stra_2$ consider the strategy $\ov{\stra}_2$
as follows: if $\ov{\stra}_1(s)=(s,A,B)$ and $\stra_2(s)=b \in 
\OptSelCount(v,s,\stra_1(s))$, then at state $(s,A,B)$ choose
the successor $(s,A,b)$. 
Since $\ov{\stra}_1$ is an almost-sure winning strategy for 
$\Safe(\ov{F})$, it follows that in the Markov chain obtained by 
fixing $\ov{\stra}_1$ and $\ov{\stra}_2$ in $\ov{G}_v$, 
all closed connected recurrent set of states that intersect 
with $\ov{A}$ are contained in $\ov{A}$, and from all 
states of $\ov{A}$ the closed connected recurrent set of states 
within $\ov{A}$ are reached with probability~1. 
It follows that in the Markov chain obtained from fixing 
$\stra_1$ and $\stra_2$ in $G$ 
all closed connected recurrent set of states that intersect 
with $\ov{A}\cap S $ are contained in $\ov{A} \cap S$, and from all 
states of $\ov{A}\cap S$ the closed connected recurrent set of states 
within $\ov{A} \cap S$ are reached with probability~1. 
The desired result follows. 
\qed
\end{proof}

\begin{lemma}\label{lemm:stra-improve-safe2}
Let $\gamma_i$ and $\gamma_{i+1}$ be the player-1 selectors obtained at 
iterations $i$ and $i+1$ of Algorithm~\ref{algorithm:strategy-improve-safe}.
Let $I=\set{s \in S \setminus (W_1 \cup T) \mid \Pre_1(v_i)(s) > v_i(s)}=\emptyset$,
and $(\ov{A}_i \cap S)\setminus W_1 \neq \emptyset$.  
Let $v_i=\vas{\overline{\gamma}_i}(\Safe(F))$ and 
$v_{i+1}=\vas{\overline{\gamma}_{i+1}}(\Safe(F))$.
Then 
$v_{i+1}(s) \geq v_i(s)$ for all states $s\in S$, 
and $v_{i+1}(s) > v_i(s)$ for some state $s\in (\ov{A}_i \cap S) \setminus W_1$.
\end{lemma}
\begin{proof}
We first show that $v_{i+1} \geq v_i$.
Let $U=(\ov{A}_i \cap S)\setminus W_1$.
Let $w_i(s) = 1 - v_i(s)$ for all states $s \in S$. 
Since $v_i =\vas{\overline{\gamma}_i}(\Safe(F))$, it follows that
for all states $s\in S$ and all moves $a_2 \in \mov_2(s)$, we have
\[
  w_i(s) \geq \sum_{t \in S} w_i(t) \cdot \trans_{\gamma_i}(s,a_2).  
\]
The selector $\xi_1(s)$ chosen for $\gamma_{i+1}$ at $s \in U$ satisfies that
$\xi_1(s) \in \OptSel(v_i,s)$.
It follows that for all states $s\in S$ and all moves $a_2 \in \mov_2(s)$, 
we have
\[
  w_i(s) \geq \sum_{t \in S} w_i(t) \cdot \trans_{\gamma_{i+1}}(s,a_2).  
\]
It follows that the maximal probability with which player~2 can reach
$T$ against the strategy $\ov{\gamma}_{i+1}$ is at most $w_i$.
It follows that $v_i(s) \leq v_{i+1}(s)$.

We now argue that for some state $s \in U$ we have $v_{i+1}(s)>v_i(s)$.
Given the strategy $\ov{\gamma}_{i+1}$, consider a pure memoryless
counter-optimal strategy $\stra_2$ for player~2 to reach $T$.
Since the selectors $\gamma_{i+1}(s)$ at states $s\in U$ are obtained from the 
almost-sure strategy $\ov{\stra}$ in the turn-based game $\ov{G}_{v_i}$ 
to satisfy $\Safe(\ov{F})$, 
it follows from Lemma~\ref{lemm:stra-improve-safetb} 
that if for every state $s \in U$, the action 
$\stra_2(s) \in \OptSelCount(v_i,s,\gamma_{i+1})$, then from 
all states $s \in U$, the game stays safe in $F$ with probability~1.
Since $\ov{\gamma}_{i+1}$ is a given strategy for player~1, and 
$\stra_2$ is counter-optimal against $\ov{\gamma}_{i+1}$, this 
would imply that $U\subseteq \set{s \in S \mid \va(\Safe(F))=1}$.
This would contradict that $W_1=\set{s\in S \mid \va(\Safe(F))=1}$ 
and $U \cap W_1=\emptyset$.
It follows that for some state $s^* \in U$ we have $\stra_2(s^*)
\not\in \OptSelCount(v_i,s^*,\gamma_{i+1})$,
and since $\gamma_{i+1}(s^*) \in \OptSel(v_i,s^*)$ 
we have 
\[
v_i(s^*) < \sum_{t \in S} v_i(t) \cdot \trans_{\gamma_{i+1}}(s^*,\stra_2(s^*));
\]
in other words, we have
\[
w_i(s^*) > \sum_{t \in S} w_i(t) \cdot \trans_{\gamma_{i+1}}(s^*,\stra_2(s^*)).
\]
Define a valuation $z$ as follows:
$z(s)=w_i(s)$ for $s \neq s^*$, and 
$z(s^*)=\sum_{t \in S} w_i(t) \cdot \trans_{\gamma_{i+1}}(s^*,\stra_2(s^*))$.
Hence $z < w_i$, and given the strategy $\ov{\gamma}_{i+1}$ and the
counter-optimal strategy $\stra_2$, the valuation $z$ satisfies the
inequalities of the linear-program for reachability to $T$.
It follows that the probability to reach $T$ given $\ov{\gamma}_{i+1}$ 
is at most $z$.
Since $z < w_i$, it follows that $v_{i+1}(s) \geq v_i(s)$ for all $s\in S$,
and $v_{i+1}(s^*) > v_i(s^*)$.
This concludes the proof.
\qed
\end{proof}

We obtain the following theorem from Lemma~\ref{lemm:stra-improve-safe1}
and Lemma~\ref{lemm:stra-improve-safe2} that shows that the sequences of
values we obtain is monotonically non-decreasing.

\begin{theorem}[Monotonicity of values]\label{thrm:safe-mono}
For $i\geq 0$, let $\gamma_{i}$ and $\gamma_{i+1}$ be the player-1 selectors obtained
at iterations $i$ and $i+1$ of Algorithm~\ref{algorithm:strategy-improve-safe}.
If $\gamma_{i}\neq \gamma_{i+1}$, then 
$\vas{\overline{\gamma}_i}(\Safe(F)) < \vas{\overline{\gamma}_{i+1}}(\Safe(F))$.
\end{theorem}

\begin{theorem}[Optimality on termination]\label{thrm:safe-termination}
Let $v_i$ be the valuation at iteration $i$ of 
Algorithm~\ref{algorithm:strategy-improve-safe} such that 
$v_i=\vas{\ov{\gamma}_i}(\Safe(F))$. 
If  
$I=\set{s\in S \setminus (W_1 \cup T) \mid \Pre_1(v_i)(s) > v_i(s)}=\emptyset$,
and $(\ov{A}_i \cap S)\setminus W_1=\emptyset$,
then $\ov{\gamma}_i$ is an optimal strategy and 
$v_i=\va(\Safe(F))$.
\end{theorem}
\begin{proof}
We show that for all memoryless strategies $\stra_1$ for player~1 we have 
$\vas{\stra_1}(\Safe(F)) \leq v_i$.
Since memoryless optimal strategies exist for concurrent games with safety objectives
(Theorem~\ref{thrm:memory-determinacy}) the desired result follows.

Let $\ov{\stra}_2$ be a pure memoryless optimal strategy for 
player~2 in $\ov{G}_{v_i}$ for the objective 
complementary to $\Safe(\ov{F})$, 
where $(\ov{G}_{v_i},\Safe(\ov{F}))=\TB(G,v_i,F)$.
Consider a memoryless strategy $\stra_1$ for player~1,
and we define a pure memoryless strategy $\stra_2$
for player~2 as follows.
\begin{enumerate}
\item If $\stra_1(s) \not\in \OptSel(v_i,s)$, then $\stra_2(s)=b \in \mov_2(s)$,
	such that $\Pre_{\stra_1(s),b}(v_i)(s) < v_i(s)$;
	(such a $b$ exists since $\stra_1(s) \not\in \OptSel(v_i,s)$).

\item If $\stra_1(s) \in \OptSel(v_i,s)$, then let $A=\supp(\stra_1(s))$,
	and consider $B$ such that $B=\OptSelCount(v_i,s,\stra_1(s))$.
	Then we have $\stra_2(s)=b$, such that $\ov{\stra}_2((s,A,B))=(s,A,b)$. 
\end{enumerate}
Observe that by construction of $\stra_2$, for all 
$s \in S \setminus (W_1 \cup T)$, we have 
$\Pre_{\stra_1(s),\stra_2(s)}(v_i)(s) \leq v_i(s)$.
We first show that in the Markov chain obtained by fixing $\stra_1$ and 
$\stra_2$ in $G$, there is no closed connected recurrent set of states $C$
such that $C \subseteq S \setminus (W_1 \cup T)$.
Assume towards contradiction that $C$ is a closed connected recurrent 
set of states in $S \setminus (W_1 \cup T)$.
The following case analysis achieves the contradiction.
\begin{enumerate}
\item Suppose for every state $s \in C$ we have $\stra_1(s) \in \OptSel(v_i,s)$.
Then consider the strategy $\ov{\stra}_1$ in $\ov{G}_{v_i}$ such that 
for a state $s \in C$ we have $\ov{\stra}_1(s)=(s,A,B)$,
where $\stra_1(s)=A$, and $B=\OptSelCount(v_i,s,\stra_1(s))$.
Since $C$ is closed connected recurrent states, it follows by construction 
that for all states $s \in C$ in the game $\ov{G}_{v_i}$ we have 
$\Prb_s^{\ov{\stra}_1,\ov{\stra}_2}(\Safe(\ov{C}))=1$,
where $\ov{C}=C \cup \set{(s,A,B) \mid s \in C} \cup \set{(s,A,b) \mid s \in C}$.
It follows that for all $s \in C$ in $\ov{G}_{v_i}$ 
we have $\Prb_s^{\ov{\stra}_1,\ov{\stra}_2}(\Safe(\ov{F}))=1$.
Since $\ov{\stra}_2$ is an optimal strategy, it follows that $C 
\subseteq (\ov{A}_i \cap S)\setminus W_1$.
This contradicts that $(\ov{A}_i \cap S) \setminus W_1=\emptyset$.

\item Otherwise for some state $s^* \in C$ we have $\stra_1(s^*) \not\in
\OptSel(v_i,s^*)$.
Let $r=\min\set{q \mid U_q(v_i) \cap C \neq \emptyset}$, i.e., 
$r$ is the least value-class with non-empty intersection with $C$.
Hence it follows that for all $q<r$, we have 
$U_q(v_i) \cap C=\emptyset$.
Observe that since for all $s \in C$ we have 
$\Pre_{\stra_1(s),\stra_2(s)}(v_i)(s) \leq v_i(s)$,
it follows that for all $s \in U_r(v_i)$ either
(a)~$\dest(s,\stra_1(s),\stra_2(s))\subseteq U_r(v_i)$;
or (b)~$\dest(s,\stra_1(s),\stra_2(s)) \cap U_q(v_i) \neq \emptyset$,
for some $q<r$.
Since $U_r(v_i)$ is the least value-class with non-empty intersection 
with $C$, it follows that for all $s \in U_r(v_i)$ we have 
$\dest(s,\stra_1(s),\stra_2(s)) \subseteq U_r(v_i)$.
It follows that $C \subseteq U_r(v_i)$. 
Consider the state $s^* \in C$ such that $\stra_1(s^*) \not\in \OptSel(v_i,s)$.
By the construction of $\stra_2(s)$, we have 
$\Pre_{\stra_1(s^*),\stra_2(s^*)}(v_i)(s^*) < v_i(s^*)$.
Hence we must have $\dest(s^*,\stra_1(s^*),\stra_2(s^*)) \cap U_{q}(v_i) 
\neq \emptyset$, for some $q <r$.
Thus we have a contradiction.
\end{enumerate}
It follows from above that there is no closed connected recurrent set of states
in $S\setminus (W_1 \cup T)$, and hence with probability~1 
the game reaches $W_1 \cup T$ from all states in $S \setminus (W_1 \cup T)$.
Hence the probability to satisfy $\Safe(F)$ is equal to the probability 
to reach $W_1$.
Since for all states $s \in S\setminus (W_1 \cup T)$ we have 
$\Pre_{\stra_1(s),\stra_2(s)}(v_i)(s) \leq v_i(s)$, 
it follows that given the strategies $\stra_1$ and 
$\stra_2$, the valuation $v_i$ satisfies all the inequalities 
for linear program to reach $W_1$.
It follows that the probability to reach $W_1$ from $s$ is 
atmost $v_i(s)$.
It follows that for all $s \in S\setminus (W_1 \cup T)$ 
we have $\vas{\stra_1}(\Safe(F))(s)\leq v_i(s)$.
The result follows.
\qed
\end{proof}

\medskip\noindent{\bf Convergence.}
We first observe that since pure memoryless optimal strategies exist for
turn-based stochastic games with safety objectives 
(Theorem~\ref{thrm:memory-determinacy}), for turn-based stochastic games 
it suffices to iterate over pure memoryless selectors.
Since the number of pure memoryless strategies is
finite, it follows for turn-based stochastic games 
Algorithm~\ref{algorithm:strategy-improve-safe} always 
terminates and yields an optimal strategy.
For concurrent games, we will use the result that for 
$\vare>0$, there is a \emph{$k$-uniform memoryless} strategy
that achieves the value of a safety objective with in $\vare$.
We first define $k$-uniform memoryless strategies.
A selector $\xi$ for player~1 is \emph{$k$-uniform} if for 
all $s \in S \setminus (T \cup W_1)$ and all $a \in \supp(\stra_1(s))$ 
there exists $i,j \in \Nats$ such that $0 \leq i \leq j \leq k$ and
$\xi(s)(a)=\frac{i}{j}$, i.e., the moves in the support are  played 
with probability that are multiples of $\frac{1}{\ell}$ with $\ell \leq k$.

\begin{lemma}\label{lemm:kuniform}
For all concurrent game graphs $G$, for all safety objectives
$\Safe(F)$, for $F \subseteq S$,
for all $\vare>0$, there exist $k$-uniform selectors $\xi$ such that
$\overline{\xi}$ is an $\vare$-optimal strategy for
$k=2^{\frac{2^{O(n)}}{\vare}}$, where $n=|S|$.
\end{lemma}
\begin{proof} {\em (Sketch).} 
For a rational $r$, using the results of~\cite{dAM04}, it can be 
shown that whether $\va(\Safe(F))(s)\geq r$ can be expressed
in the quantifier free fragment of the theory of reals.
Then using the formula in the theory of reals and 
Theorem~13.12 of~\cite{BasuBook}, it can be shown that if there 
is a memoryless strategy $\stra_1$ that achieves value at least 
$r$, 
then there is a $k$-uniform memoryless strategy $\stra_1^k$ 
that achieves value at least $r-\vare$,
where $k=2^{\frac{2^{O(n)}}{\vare}}$, for $n=|S|$.
\qed
\end{proof}

\noindent{\bf Strategy improvement with $k$-uniform selectors.}
We first argue that if we restrict 
Algorithm~\ref{algorithm:strategy-improve-safe} 
such that every iteration yields a $k$-uniform selector, then 
the algorithm terminates.
If we  restrict to $k$-uniform selectors, then a concurrent game graph
$G$ can be converted to a turn-based stochastic game graph,
where player~1 first chooses a $k$-uniform selector, then player~2
chooses an action, and then the transition is determined by the
chosen $k$-uniform selector of player~1, the action of player~2
and the transition function $\trans$ of the game graph $G$.
Then by termination of turn-based stochastic games it follows
that the algorithm will terminate.
Given $k$, let us denote by $z_i^k$ the valuation of 
Algorithm~\ref{algorithm:strategy-improve-safe} at iteration $i$,
where the selectors are restricted to be $k$-uniform,
and $v_i$ is the valuation of Algorithm~\ref{algorithm:strategy-improve-safe}
at iteration $i$.
Since $v_i$ is obtained without any restriction, it follows that
for all $k>0$, for all $i \geq 0$, we have 
$z_i^k \leq v_i$.
From Lemma~\ref{lemm:kuniform} it follows that for all $\vare>0$,
there exists a $k>0$ and $i\geq 0$ such that for all $s$ we have 
$z_i^k(s) \geq \va(\Safe(F))(s) -\vare$.
This gives us the following result.

\begin{theorem}[Convergence]
Let $v_i$ be the valuation obtained at iteration $i$ 
of Algorithm~\ref{algorithm:strategy-improve-safe}.
Then the following assertions hold.
\begin{enumerate}
\item For all $\vare>0$, there exists $i$ such 
that for all $s$ we have 
$v_i(s) \geq \va(\Safe(F))(s) -\vare$.

\item $\lim_{i \to \infty} v_i =\va(\Safe(F))$.

\end{enumerate}
\end{theorem}

\noindent{\bf Complexity.} 
Algorithm~\ref{algorithm:strategy-improve-safe} may not terminate
in general. We briefly describe the complexity of every iteration.
Given a valuation $v_i$, the computation of $\Pre_1(v_i)$ 
involves solution of matrix games with rewards $v_i$ and 
can be computed in polynomial time using linear-programming.
Given $v_i$ and $\Pre_1(v_i)=v_i$, 
the set $\OptSel(v_i,s)$ and $\OptSelCount(v_i,s)$ 
can be computed by enumerating the subsets of available actions
at $s$ and then using linear-programming: for example to check 
$(A,B)\in \OptSelCount(v_i,s)$ it suffices to check that 
there is an selector $\xi_1$ such that $\xi_1$ is optimal
(i.e. for all actions $b \in \mov_2(s)$ we have 
$\Pre_{\xi_1,b}(v_i)(s)\geq v_i(s)$);
for all $a \in A$ we have $\xi_1(a)>0$, and for all
$a \not \in A$ we have $\xi_1(a)=0$; and to check $B$ 
is the set of counter-optimal actions we check 
that for $b \in B$ we have 
$\Pre_{\xi_1,b}(v_i)(s)= v_i(s)$;
and for $b \not\in B$ we have 
$\Pre_{\xi_1,b}(v_i)(s)> v_i(s)$.
All the above can be solved by checking feasibility of a 
set of linear inequalities.
Hence $\TB(G,v_i,F)$ can be computed in time 
polynomial in size of $G$ and $v_i$ and exponential in the 
number of moves.
The set of almost-sure winning states in turn-based stochastic
games with safety objectives can be computed in 
linear-time~\cite{crg-tcs07}.

\section{Termination for Approximation and Turn-based Games}
In this section we present termination criteria for strategy improvement
algorithms for concurrent games for $\vare$-approximation, and 
then present an improved termination condition for turn-based games.

\medskip\noindent{\bf Termination for concurrent games.}
A strategy improvement algorithm for reachability games 
was presented in~\cite{CdAH06}.
We refer to the algorithm of~\cite{CdAH06} as the 
\emph{reachability strategy improvement algorithm}.
The reachability strategy improvement algorithm is simpler
than Algorithm~\ref{algorithm:strategy-improve-safe}: it is similar to 
Algorithm~\ref{algorithm:strategy-improve-safe} and in every iteration
only Step~3.2 is executed (and Step 3.3 need not be executed).
Applying the reachability strategy improvement algorithm of~\cite{CdAH06} 
for player~2, for a reachability objective $\Reach(T)$, we obtain a
sequence of valuations $(u_i)_{i\geq 0}$ such that 
(a) $u_{i+1} \geq u_i$;
(b) if $u_{i+1}=u_i$, then $u_i=\vb(\Reach(T))$; and
(c) $\lim_{i \to \infty} u_i =\vb(\Reach(T))$.
Given a concurrent game $G$ with $F \subs S$ and $T=S\setminus F$,
we apply the reachability strategy improvement algorithm to obtain
the sequence of valuation $(u_i)_{i\geq 0}$ as above, and 
we apply Algorithm~\ref{algorithm:strategy-improve-safe} 
to obtain a sequence of valuation $(v_i)_{i \geq 0}$.
The termination criteria are as follows:
\begin{enumerate}
\item if for some $i$ we have $u_{i+1}=u_i$, then we have 
$u_i=\vb(\Reach(T))$, and $1-u_i=\va(\Safe(F))$, and we 
obtain the values of the game;
\item if for some $i$ we have $v_{i+1}=v_i$, then we have 
$1-v_i=\vb(\Reach(T))$, and $v_i=\va(\Safe(F))$, and we 
obtain the values of the game; and
\item for $\vare>0$, if for some $i\geq 0$, we have $u_i + v_i \geq 1-\vare$,
then for all $s \in S$ we have
$v_i(s)\geq \va(\Safe(F))(s) -\vare$ and
$u_i(s)\geq \vb(\Reach(T))(s) -\vare$ (i.e., the algorithm
can stop for $\vare$-approximation).
\end{enumerate}
Observe that since $(u_i)_{i\geq 0}$ and $(v_i)_{i \geq 0}$ are 
both monotonically non-decreasing and $\va(\Safe(F))+ \vb(\Reach(T))=1$, 
it follows that if $u_i + v_i \geq 1-\vare$, then forall 
$j\geq i$ we have $u_i \geq u_j -\vare$ and
$v_i \geq v_j -\vare$.
This establishes that $u_i \geq \va(\Safe(F)) -\vare$ and
$v_i \geq \vb(\Reach(T)) -\vare$;
and the correctness of the stopping criteria (3) for 
$\vare$-approximation follows.
We also note that instead of applying the reachability 
strategy improvement algorithm, a value-iteration algorithm
can be applied for reachability games to obtain a 
sequence of valuation with properties similar to $(u_i)_{i \geq 0}$
and the above termination criteria can be applied.

\begin{theorem}
Let $G$ be a concurrent game graph with a safety objective $\Safe(F)$.
Algorithm~\ref{algorithm:strategy-improve-safe} and the
reachability strategy improvement algorithm for player~2 for the 
reachability objective $\Reach(S\setminus F)$ yield  sequence of valuations 
$(v_i)_{i \geq 0}$ and $(u_i)_{i \geq 0}$, respectively, such that 
(a)~for all $i\geq 0$, we have $v_i \leq \va(\Safe(F)) \leq 1-u_i$; 
and
(b)~$\lim_{i \to \infty} v_i = \lim_{i \to \infty} 1-u_i = \va(\Safe(F))$. 
\end{theorem}

\medskip\noindent{\bf Termination for turn-based games.}
For turn-based stochastic games 
Algorithm~\ref{algorithm:strategy-improve-safe} and as well
as the reachability strategy improvement algorithm terminates.
Each iteration of the reachability strategy improvement algorithm 
of~\cite{CdAH06} is computable in polynomial time, and here we present a 
termination guarantee for the reachability strategy improvement algorithm.
To apply the reachability strategy improvement algorithm 
we assume the objective of player~1 to be a reachability 
objective $\Reach(T)$, and the correctness of the algorithm
relies on the notion of \emph{proper strategies}.
Let $W_2=\set{s\in S \mid \va(\Reach(T))(s)=0}$.
Then the notion of proper strategies and its properties are
as follows.

\begin{definition}[Proper strategies and selectors]
A player-1 strategy $\stra_1$ is \emph{proper}
if for all player-2 strategies $\stra_2$,
and for all states $s \in S \setminus \wab$, we have 
$\Prb_s^{\stra_1,\stra_2}(\Reach(T\cup W_2))=1$.
A player-1 selector $\xi_1$ is \emph{proper} if the memoryless player-1
strategy $\overline{\xi}_1$ is proper. 
\end{definition} 

\begin{lemma}[\cite{CdAH06}]
Given a selector $\xi_1$ for player~1, the memoryless player-1 strategy
$\overline{\xi}_1$ is proper iff for every pure selector $\xi_2$ for
player~2, and for all states $s \in S$, we have 
$\Prb_s^{\overline{\xi}_1, \overline{\xi}_2} (\Reach(T \cup W_2))=1$.
\end{lemma}

The following result follows from the result of~\cite{CdAH06}
specialized for the case of turn-based stochastic games.

\begin{lemma}
Let $G$ be a turn-based stochastic game with reachability objective
$\Reach(T)$ for player~1.
Let $\gamma_0$ be the initial selector, and $\gamma_i$ be the 
selector obtained at iteration $i$ of the reachability 
strategy improvement algorithm.
If $\gamma_i$ is a pure, proper selector, then the following 
assertions hold:
\begin{enumerate}
\item for all $i\geq 0$, we have $\gamma_i$ is a pure, proper selector;
\item for all $i\geq 0$, we have $u_{i+1} \geq u_i$, where 
$u_i=\vas{\ov{\gamma}_i}(\Reach(T))$ and 
$u_{i+1}=\vas{\ov{\gamma}_{i+1}}(\Reach(T))$; and
\item if $u_{i+1}=u_i$, then $u_i=\va(\Reach(T))$, and
there exists $i$ such that $u_{i+1}=u_i$.
\end{enumerate}
\end{lemma}

The strategy improvement algorithm of Condon~\cite{Con93} works
only for \emph{halting games}, but the reachability 
strategy improvement algorithm works if we start with a pure, 
proper selector for reachability games that are not halting.
Hence to use the reachability strategy improvement algorithm to
compute values we need to start with a pure, proper selector.
We present a procedure to compute a pure, proper selector,
and then present termination bounds (i.e., bounds on $i$ such 
that $u_{i+1}=u_i$).
The construction of pure, proper selector is based on the notion 
of \emph{attractors} defined below.

\medskip\noindent{\em Attractor strategy.}
Let $A_0=W_2 \cup T$,  and for $i\geq 0$ we have
\[
A_{i+1}= A_i \cup \set{s \in S_1 \cup S_R \mid E(s) \cap A_i \neq \emptyset}
\cup \set{s \in S_2 \mid E(s) \subs A_i}.
\]
Since for all $s \in S \setminus W_2$ we have $\va(\Reach(T))>0$,
it follows that from all states in $S\setminus W_2$ player~1
can ensure that $T$ is reached with positive probability.
It follows that for some $i \geq 0$ we have $A_i=S$.
The pure \emph{attractor} selector $\xi^*$ is as follows:
for a state $s \in (A_{i+1}\setminus A_i) \cap S_1$ 
we have $\xi^*(s)(t)=1$, where $t \in A_i$ (such a $t$ 
exists by construction).
The pure memoryless strategy $\ov{\xi^*}$ ensures that for 
all $i\geq 0$, from $A_{i+1}$ the game reaches $A_i$ 
with positive probability.
Hence there is no end-component $C$ contained in 
$S\setminus (W_2 \cup T)$ in the MDP $G_{\ov{\xi^*}}$.
It follows that $\xi^*$ is a pure selector that is proper,
and the selector $\xi^*$ can be computed in $O(|E|)$ time.
This completes the reachability strategy 
improvement algorithm for turn-based stochastic 
games.
We now present the termination bounds.

\medskip\noindent{\em Termination bounds.}
We present termination bounds for binary turn-based 
stochastic games.
A turn-based stochastic game is binary if for all $s\in S_R$
we have $|E(s)|\leq 2$, and for all $s \in S_R$ if $|E(s)|=2$,
then for all $t \in E(s)$ we have $\trans(s)(t)=\frac{1}{2}$,
i.e., for all probabilistic states there are at most two
successors and the transition function $\trans$ is uniform.

\begin{lemma}\label{lemm:MC-bound}
Let $G$ be a binary Markov chain with $|S|$ states with a reachability
objective $\Reach(T)$.
Then for all $s \in S$ we have 
$\va(\Reach(T))=\frac{p}{q}$, with $p,q \in \nats$ and $p,q \leq 4^{|S|-1}$. 
\end{lemma}
\begin{proof}
The results follow as a special case of Lemma~2 of~\cite{Con93}.
Lemma~2 of~\cite{Con93} holds for halting turn-based stochastic games,
and since Markov chains reaches the set of closed connected recurrent
states with probability~1 from all states the result follows.
\qed
\end{proof}

\begin{lemma}\label{lemm:TB-bound}
Let $G$ be a binary turn-based stochastic game with a reachability
objective $\Reach(T)$.
Then for all $s \in S$ we have 
$\va(\Reach(T))=\frac{p}{q}$, with $p,q \in \nats$ and $p,q \leq 4^{|S_R|-1}$. 
\end{lemma}
\begin{proof}
Since pure memoryless optimal strategies exist for both players 
(Theorem~\ref{thrm:memory-determinacy}),
we fix pure memoryless optimal strategies $\stra_1$ and 
$\stra_2$ for both players.
The Markov chain $G_{\stra_1,\stra_2}$ can be then reduced to an equivalent 
Markov chains with $|S_R|$ states (since we fix deterministic successors
for states in $S_1 \cup S_2$, they can be collapsed to their successors).
The result then follows from Lemma~\ref{lemm:MC-bound}.
\qed
\end{proof}

From Lemma~\ref{lemm:TB-bound} it follows that at iteration~$i$ of the
reachability strategy improvement algorithm either 
the sum of the values either increases by $\frac{1}{4^{|S_R|-1}}$ or else
there is a valuation $u_i$ such that $u_{i+1}=u_i$.
Since the sum of values of all states can be at most $|S|$, it follows
that algorithm terminates in at most $|S| \cdot 4^{|S_R|-1}$ steps.
Moreover, since the number of pure memoryless strategies is at most
$\prod_{s \in S_1} |E(s)|$, the algorithm terminates in at most
$\prod_{s \in S_1} |E(s)|$ steps.
It follows from the results of~\cite{ZP95} that a turn-based stochastic
game graph $G$ can be reduced to a equivalent binary turn-based
stochastic game graph $G'$ such that the set of player~1 and player~2
states in $G$ and $G'$ are the same and the number of probabilistic
states in $G'$ is $O(|\trans|)$, where $|\trans|$ is the size of the
transition function in $G$.
Thus we obtain the following result.

\begin{theorem}
Let $G$ be a turn-based stochastic game with a reachability objective 
$\Reach(T)$, then the reachability strategy improvement algorithm
computes the values in time 
\[
O\big(\min\set{\prod_{s \in S_1} |E(s)|, 2^{O(|\trans|)} } \cdot \mathit{poly}(|G|\big);
\]
where $\mathit{poly}$ is polynomial function.
\end{theorem} 

The results of~\cite{GH07} presented an algorithm for turn-based 
stochastic games that works in time $O(|S_R| ! \cdot \mathit{poly}(|G|))$.
The algorithm of~\cite{GH07} works only for turn-based stochastic
games, for general turn-based stochastic games the complexity of 
the algorithm of~\cite{GH07} is better.
However, for turn-based stochastic games where the transition function 
at all states can expressed in constant bits we have 
$|\trans| =O(|S_R|)$.
In these cases the reachability strategy improvement algorithm (that works 
for both concurrent and turn-based stochastic games) 
works in time $2^{O(|S_R|)} \cdot \mathit{poly}(|G|)$ 
as compared to the time $2^{O(|S_R|\cdot \log(|S_R|)} \cdot \mathit{poly}(|G|)$
of  the algorithm of~\cite{GH07}.

\end{document}

%% file: example1.eepic
\setlength{\unitlength}{0.00033333in}
\begingroup\makeatletter\ifx\SetFigFont\undefined%
\gdef\SetFigFont#1#2#3#4#5{%
  \reset@font\fontsize{#1}{#2pt}%
  \fontfamily{#3}\fontseries{#4}\fontshape{#5}%
  \selectfont}%
\fi\endgroup%
{\renewcommand{\dashlinestretch}{30}
\begin{picture}(6880,2159)(0,-10)
\put(6304,952){\makebox(0,0)[lb]{{\SetFigFont{8}{9.6}{\rmdefault}{\mddefault}{\updefault}2/3}}}
\thicklines
\put(6603.464,359.500){\arc{521.746}{4.0303}{8.5361}}
\blacken\path(6662.746,667.536)(6439.000,562.000)(6686.087,549.828)(6662.746,667.536)
\put(724,1822){\ellipse{600}{600}}
\put(6124,1822){\ellipse{600}{600}}
\path(4624,2122)(4324,1822)(4624,1522)
	(4924,1822)(4624,2122)
\path(1924,2122)(2524,2122)(2524,1522)
	(1924,1522)(1924,2122)
\path(424,622)(1024,622)(1024,22)
	(424,22)(424,622)
\path(5824,622)(6424,622)(6424,22)
	(5824,22)(5824,622)
\path(1924,1822)(1024,1822)
\blacken\path(1264.000,1882.000)(1024.000,1822.000)(1264.000,1762.000)(1264.000,1882.000)
\path(4924,1822)(5824,1822)
\blacken\path(5584.000,1762.000)(5824.000,1822.000)(5584.000,1882.000)(5584.000,1762.000)
\path(724,1522)(724,622)
\blacken\path(664.000,862.000)(724.000,622.000)(784.000,862.000)(664.000,862.000)
\path(724,1522)(5824,322)
\blacken\path(5576.638,318.564)(5824.000,322.000)(5604.122,435.374)(5576.638,318.564)
\path(6124,1522)(1024,322)
\blacken\path(1243.878,435.374)(1024.000,322.000)(1271.362,318.564)(1243.878,435.374)
\path(6124,1522)(6124,622)
\blacken\path(6064.000,862.000)(6124.000,622.000)(6184.000,862.000)(6064.000,862.000)
\path(2524,1957)(4399,1957)
\blacken\path(4159.000,1897.000)(4399.000,1957.000)(4159.000,2017.000)(4159.000,1897.000)
\path(4339,1717)(2539,1717)
\blacken\path(2779.000,1777.000)(2539.000,1717.000)(2779.000,1657.000)(2779.000,1777.000)
\put(1999,1732){\makebox(0,0)[lb]{{\SetFigFont{8}{9.6}{\rmdefault}{\mddefault}{\updefault}$s_0$}}}
\put(4429,1717){\makebox(0,0)[lb]{{\SetFigFont{8}{9.6}{\rmdefault}{\mddefault}{\updefault}$s_1$}}}
\put(6004,1702){\makebox(0,0)[lb]{{\SetFigFont{8}{9.6}{\rmdefault}{\mddefault}{\updefault}$s_3$}}}
\put(574,1732){\makebox(0,0)[lb]{{\SetFigFont{8}{9.6}{\rmdefault}{\mddefault}{\updefault}$s_2$}}}
\put(499,202){\makebox(0,0)[lb]{{\SetFigFont{8}{9.6}{\rmdefault}{\mddefault}{\updefault}$s_6$}}}
\put(5929,202){\makebox(0,0)[lb]{{\SetFigFont{8}{9.6}{\rmdefault}{\mddefault}{\updefault}$s_5$}}}
\put(169,937){\makebox(0,0)[lb]{{\SetFigFont{8}{9.6}{\rmdefault}{\mddefault}{\updefault}2/3}}}
\put(1444,937){\makebox(0,0)[lb]{{\SetFigFont{8}{9.6}{\rmdefault}{\mddefault}{\updefault}1/3}}}
\put(4939,937){\makebox(0,0)[lb]{{\SetFigFont{8}{9.6}{\rmdefault}{\mddefault}{\updefault}1/3}}}
\put(260.205,354.411){\arc{486.746}{1.0609}{5.3702}}
\blacken\path(163.173,519.269)(409.000,547.000)(179.044,638.215)(163.173,519.269)
\end{picture}
}